\documentclass[10pt, twoside, twocolumn]{IEEEtran}

\usepackage[latin1]{inputenc}
\usepackage[T1]{fontenc}
\usepackage{amssymb,mathtools}   
\usepackage{amsthm}
\usepackage{graphicx}
\usepackage{psfrag}
\usepackage{color}

\usepackage{adjustbox}
\usepackage{suffix}
\usepackage{booktabs}  % tabla
\usepackage{multirow} %tabla
\usepackage{relsize}

\newtheorem{theorem}{Theorem}
\newtheorem{lemma}[theorem]{Lemma}

\usepackage[numbers,sort&compress]{natbib}

\DeclarePairedDelimiter\abs{\lvert}{\rvert} %para abs
\usepackage{float}
\usepackage{stfloats}
\usepackage{textcomp}

\begin{document}

%\title{\LARGE \bf A Generalized Distribution for the Ratio of Independent and Arbitrary $\alpha$-$\mu$ Random Variables with Applications}

%\title{\LARGE \bf On the Statistics of the Ratio of Independent and Arbitrary $\alpha$-$\mu$ Random Variables: a General Framework and Applications}

\title{On the Secrecy Performance Over $\mathrm{N}$-Wave with Diffuse
Power Fading Channel}

% * <edgar.benitez.o@gmail.com> 2018-12-08T16:12:48.156Z:
%
% ^.

\author{Jos\'e~David~Vega~S\'anchez, D.~P.~Moya~Osorio, F. Javier L\'opez-Mart\'inez, Martha Cecilia Paredes, and Luis~Urquiza-Aguiar
\thanks{Jos\'e~David~Vega~S\'anchez, Martha Cecilia Paredes, and Luis~Urquiza-Aguiar are with the  
Departamento de Electr\'onica, Telecomunicaciones y Redes de Informaci\'on, Escuela Polit\'ecnica Nacional (EPN),
Quito,  170525, Ecuador. (e-mail: jose.vega01@epn.edu.ec; cecilia.paredes@epn.edu.ec; luis.urquiza@epn.edu.ec)}
\thanks{D.~P.~Moya~Osorio is with the Centre for Wireless Communications (CWC), University of Oulu, Finland. (e-mail: diana.moyaosorio@oulu.fi)}
\thanks{ F. Javier L\'opez-Mart\'inez is with Departamento de Ingenier\'ia de Comunicaciones, Universidad de M\'alaga - Campus de Excelencia Internacional Anadaluc\'ia Tech., M\'alaga 29071, Spain. (e-mail: fjlopezm@ic.uma.es).}
}
\maketitle
%\(\)\thispagestyle{empty}
%\pagestyle{empty}

%% ABSTRACT
\begin{abstract}
We investigate the effect of considering realistic propagation conditions different from classical Rice and Rayleigh fading on wireless physical layer security. Specifically, we study how the superposition of a number of dominant specular waves and diffusely propagating components impacts the achievable secrecy performance compared to conventional assumptions relying on the central limit theorem. We derive analytical expressions for the secrecy outage probability, which have similar complexity to other alternatives in the literature derived for simpler fading models. We provide very useful insights on the impact on physical layer security of $(i)$ the number; $(ii)$ the relative amplitudes and $(iii)$ the overall power of the dominant specular components. We show that it is possible to obtain remarkable improvements on the system secrecy performance when: $(a)$ the relative amplitudes of the dominant specular components for the eavesdropper channel are sufficiently large compared to those of the eavesdropper's channel eavesdropper, and $(b)$ the power of Bob's dominant components is significantly larger than the power of the Eve's dominant components. %The validity of the proposed expressions is confirmed via Monte Carlo simulations.
\end{abstract}


\section{Introduction}

\IEEEPARstart{T}{he} fifth-generation (5G) of mobile networks aims to raise the capacity and performance of communication systems to unprecedented levels, including ultra-high data rates, ultra-wide radio coverage, massive simultaneously connected devices and ultra-low latencies. The new scenarios of wireless systems under the umbrella of 5G include mm-Wave, device-to-device, machine-type, and vehicular communications, among many others~\cite{Liu}. In particular, recent investigations have shown that none of the well-established fading models (e.g., Rayleigh, Rician and Nakagami-$m$) present accurate fit with field measurements in mm-Wave communications~\cite{Samimi}. One of the reasons for such mismatch relies in the fact that classical fading models heavily rely on the central limit theorem (CLT), which assumes a sufficiently large number of multipath waves arriving at the receiver ends -- and such conditions are not always met \cite{Taricco}.

In the last years, some efforts have been oriented to formulate more accurate channel models that overcome such limitation. Among them, stochastic fading models that explicitly discern between the individual multipath waves classically regarded as line-of-sight (LOS) components have been proposed as a way for bridging the gap between CLT-based approaches and purely ray-based models. Durgin's two-wave with diffuse power (TWDP) \cite{Durgin} and its generalization in \cite{Romero} are known to improve the fit to field measurements in different scenarios including mm-Wave set-ups~\cite{Samimi,Paris}, compared to conventional small-scale fading models.

On the other hand, a myriad of challenges must still be overcome so that 5G converges into a reliable, safe and efficient system. One of the most critical aspects is related to the security of information transmission, given that 5G is designed to support rather diverse applications. As a consequence, highly confidential and vulnerable data is expected to be transmitted in future 5G and beyond networks, which because of their wireless nature are sensitive to eavesdropping. Regarding this, physical layer security (PLS) \cite{Barros} emerges as a promising solution to complement traditional security systems by taking advantage of the random nature of wireless channels to provide lightweight and efficient solutions for increasing the security level in some applications~\cite{Vega,Moya}. 

%% JAVI: No veo mal el siguiente párrafo completo, pero es posible que algún revisor pida que comparemos nuestros resultados con alguna de las referencias. Al ser muy específicas de mmwave, es probable que sea complicado de comparar peras con manzanas.
%% Además, aunque es muy deseable una introducción "thorough", me da la sensación de que actualmente es demasiado larga comparado con la longitud final del paper. Por ello, me he tomado la licencia de comprimir en parte el bloque donde discutes secrecy en mmwave y el posterior donde se estudia el efecto del fading en PLS. Podemos modificarlo de nuevo si te parece.
%% Ten en cuenta además que aunque nos beneficia usar el argumento de que los modelos "clásicos" de fading no tienen sentido en mmWave, no podemos menospreciar los que usan el modelo Saleh-Valenzuela como en "Spatially Sparse Precoding in Millimeter Wave MIMO Systems" ya que es un paper con 1700 citas, y el modelo Saleh-Valenzuela es en parte un modelo de rayos.
%% Te propongo cambiar el enfoque de la intro, ya que no me siento cómodo vendiendo el paper como un paper de 5G y mmwave, sino que prefiero un enfoque más fundamental centrado en ¿por qué los rayos son importantes y qué pasa cuando los trato de manera individualizada?

The physical layer security performance in wireless channels is a rather well-investigated topic in the literature. Nevertheless, because of the intricate nature of physically-motivated wireless fading models, available results are restricted to some special cases~\cite{TWDP,FTR} on which only two specular components are considered. Very recently, it was suggested in~\cite{Espinosa} that the inability to achieve perfect secrecy in wireless channels was an artifact due to the consideration of the CLT assumption. Hence, the impact of the number of multipath waves arriving at the receiver ends, as well as their relative amplitudes, has a dramatical effect on the secrecy performance. However, the results in~\cite{Espinosa} considered only the limit case of a total absence of diffuse fading, and the derivation of analytical expressions for the secrecy performance metrics was not possible for a an arbitrary number of waves. 

In this paper, we investigate the PLS performance in a wireless set-up, by assuming that the received signal is built from the superposition of an arbitrary number $\mathrm{N}$ of dominant multipath waves, plus some additional diffuse components -- this will be referred to as $\mathrm{N}$-wave with diffuse power (NWDP) fading, for which some formulations have been recently proposed in order to deal with its rather unwieldy nature \cite{Chun,Juanma}. Our goal is to perform a fine-grain characterization of the role of individual multipath waves on the secrecy performance, and to support our findings with analytical results. We derive exact expressions for the secrecy outage probability (SOP) for an arbitrary number of dominant waves at the desired and eavesdropping ends, as well as simplified approximations that become asymptotically tight in the high signal-to-noise ratio (SNR) regime. Some useful insights for improving the secrecy performance in this scenario will be derived, which are inherently linked to the underlying propagation mechanisms and characteristics captured by the NWDP fading model. The main contributions of this paper are as follows:

\begin{itemize}
\item An exact closed-form expression for the SOP in terms of well-know functions for the classical Wyner's channel model under NWDP fading. 
\item A high SNR approximation of the SOP is also derived, which can be used straightforwardly in the context of PLS. The merits of such expression are: $(i)$ when $\overline{\gamma}_\mathrm{E}$ is at high SNR regime, it is highly tight to the exact SOP; $(ii)$ it notably reduces the computational effort concerning the exact SOP. This fact helps to the wireless system designers when requiring quick evaluation of security risks.  
\item Some useful insights into the system are also provided through 
the asymptotic analysis based on the exact analytical expression of the SOP.

%\item Some useful insights regarding the secrecy performance behaviour are also provided with respect to the parameters and characteristics of the NWDP fading model. {\color{red} three aspects: the relative amplitudes} {\color{blue} of the? (recuerda que hasta este punto no has hablado profundamente de parametros, entonces no van a entender de que amplitudes estas hablando) yo sugiero dejar en terinos generales, mas caso quieras colocar los 3 aspectos, tjienes que ser mas especifico}{\color{red}, the number and the power of the dominant components for both the main and the wiretap channels. }
\end{itemize}

The remainder of this paper is organized as follows. System and channel models are described in Section II. Section III derives closed-form expressions for $(i)$ the SOP; $(ii)$ a high SNR regime of the SOP; $(iii)$ the asymptotic behaviour of the SOP over NWDP fading channel. %{\color{green}and a lower bound of the SOP is also obtained}.
Section IV shows illustrative numerical results and discussions. Finally, concluding remarks are provided in Section V.
%%%%%%%%%%%%%%%%%%%%%%%%%%%%%%%%
% Notation and terminology

\emph{Notation}: Throughout this paper, $f_{(Z)}(z)$ and $F_{(Z)}(z)$ denote the probability density function (PDF) and the cumulative distribution function (CDF) of a random variable $Z$. $\mathbb{E} \left [ \cdot \right ]$ denotes expectation,  $\Pr\left \{ \cdot  \right \}$ denotes probability, and $\abs{\cdot}$ denotes the absolute value. In addition, $L_n(\cdot)$ denotes the Laguerre polynomial~\cite[Eq.~(22.2.13)]{Abramowitz},  $\Gamma(\cdot)$ denotes the gamma function~\cite[Eq.~(6.1.1)]{Abramowitz}; $\gamma(\cdot,\cdot)$, the lower incomplete gamma
function~\cite[Eq.~(6.5.2)]{Abramowitz}; ${}_2F_1\left(\cdot,\cdot;\cdot;\cdot\right)$, denotes the hypergeometric function~\cite[Eq.~(15.1.1)]{Abramowitz}, and $\left(\cdot\right)_{\left(\cdot\right)}$ is the Pochhammer symbol~\cite[Eq.~(6.1.222)]{Abramowitz}.

%We consider the NWDP channel model, where the received signal consists in the superposition of a set of $N$ specular components with random phases plus a diffuse component whose magnitude follows Rayleigh distribution, as expressed in~\cite{Chun} 

%We also use $\mathrm{j}=\sqrt[]{-1}$ for the imaginary unit; $\approx$ to denote ``approximately equal~to''.

%%%%%%%%%%%%%%%%%%%%%%%%%%%%%%%%
%Statistics
\section{System Model
 }

%\subsection{System Model}

We consider the classic Wyner's wiretap channel as depicted in Fig.~\ref{sistema1}, where a legitimate transmitter Alice ($\mathrm{A}$) sends confidential messages to the legitimate receiver Bob ($\mathrm{B}$) through the main channel, while the eavesdropper Eve ($\mathrm{E}$) tries to intercept these messages from its received signal over the eavesdropper channel. It is assumed that the main and eavesdropper channels experience independent quasi-static fading. %Both channels experience ergodic block fading, thus channel coefficients remain constant during a block period. 
Without loss of generality, we assume that all nodes are equipped with a single antenna.

% figure sistema 1
\begin{figure}[t]
\centering 
\psfrag{A}[Bc][Bc][0.8]{$\mathrm{A}$}
\psfrag{B}[Bc][Bc][0.8]{$\mathrm{B}$}
\psfrag{E}[Bc][Bc][0.8]{$\mathrm{E}$}
\psfrag{U}[Bc][Bc][0.8]{$h_{\mathrm{AB}}$}
\psfrag{w}[Bc][Bc][0.8][-16]{$h_{\mathrm{AE}}$}
\psfrag{Main channel}[Bc][Bc][0.6]{Main channel} 
\psfrag{Wiretap channel}[Bc][Bc][0.6]{Wiretap channel} 
\includegraphics[width=0.7\linewidth]{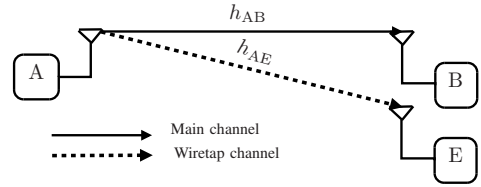} \caption{Wiretap channel consisting of a legitimate pair and one eavesdropper.}
\label{sistema1}
\end{figure}
%\subsection{Channel Model}
We express the signal at each of the receiving ends as a superposition of $\mathrm{N}$ multipath waves arising from dominant specular reflections, and $M$ additional waves associated to diffuse scattering:
\begin{equation}
\label{eq:1}
    \mathrm{R} \exp\left( {\mathrm{j}\theta} \right)=\sum_{i=1}^{\mathrm{N}} V_i \exp\left({\mathrm{j}\theta_i}\right)+\sum_{k=1}^{M} V_k \exp\left({\mathrm{j}\theta_k} \right)
\end{equation}

Because each diffuse scatterer is able to generate several multipath waves \cite{Durginb}, we can safely assume that $M\rightarrow\infty$ and hence the diffuse component{\footnote{We note that the consideration of arbitrary $N$ in \eqref{eq:1} allows for individually accounting for the effect of having multiple specular waves and largely differs from the conventional assumptions in fading modeling, reducing only for $N=0,1,2$ to the Rayleigh, Rician and TWDP cases, respectively \cite{Durgin}.
}} tends to be Gaussian distributed, i.e., $\sum_{k=1}^{M} V_k \exp\left({\mathrm{j}\theta_k}\right)\approx V_d\exp\left({\mathrm{j}\theta_d}\right)$, so that $V_d$ is Rayleigh distributed with $\mathbb{E}\{|V_d|^2\}=2\sigma^2=\Omega$. 

%where the received signal consists in the superposition of a set of $N$ specular components with random phases plus a diffuse component whose magnitude follows Rayleigh distribution, as expressed in~\cite{Chun} as
 %\begin{equation}
 %       \mathrm{NWDP}=\begin{cases}
 %      \mathrm{R} \exp\left ( {\mathrm{j}\theta} \right )=\sum_{i=1}^{N} V_i \exp\left ( {\mathrm{j}\theta_i} \right )+V_d\exp\left ( {\mathrm{j}\theta_d} \right ),\\
 %       f_{V_d}(x)=\frac{x}{\sigma^2}\exp\left ( -\frac{x^2}{2\sigma^2} \right ),
 %        \end{cases}
%\end{equation}
In \eqref{eq:1}, $V_i \exp(\mathrm{j}\theta_i)$ denotes the $i$-$th$ specular component having a constant amplitude $V_i$ and a uniformly distributed random phase $\theta_i \sim \mathcal{U}[0, 2\pi) $. The random phases for each dominant wave
are assumed to be statistically independent.% In addition, $ \Omega=\mathbb{E}  \left [ \abs{V_d}^2 \right ]=2\sigma^2$ represents  the average power of the diffuse component. Note that, for $N=0$,~\eqref{eq:1} reduces to the Rayleigh fading model (i.e., no specular component is present), for $N=1$ to the Rician fading model (i.e., a single dominant specular component), and for $N=2$ to the popular TWDP.

Let $\gamma \stackrel{\Delta}{=} \gamma_0 \mathrm{R}^2$ be the instantaneous received SNR through, where $\gamma_0\stackrel{\Delta}{=} P_{\mathrm{T}}/N_0$ is defined as the transmit SNR, with $P_{\mathrm{T}}$ being the transmit power and $N_0$ being the mean power of the additive white Gaussian noise.  Note that $\gamma$ can also be redefined for the sake of convenience as $\gamma=\overline\gamma|h|^2$, where $h$ denotes any normalized fading channel, i.e., $\mathbb{E}\{|h|^2\}=1$ and $\overline{\gamma}$ representing the average receive SNR. According to the formulation in~\cite{Chun}, the corresponding PDF and CDF of $\gamma$ over NWDP fading channel are:   
\begin{subequations}
\label{eq:SNR}
\begin{align}
\label{eq:2}
f_i(\gamma_i)=\frac{1}{\overline{\gamma}_i}\exp\left ( -\frac{\gamma_i}{\overline{\gamma}_i} \right )\sum_{n_i=0}^{\infty}C_{n_i}L_{n_i}\left ( \frac{\gamma_i}{\overline{\gamma}_i} \right ),
\end{align}
 \begin{align}\label{eq:3}
F_i(\gamma_i)=\sum_{n_i=0}^{\infty}C_{n_i}\sum_{k_i=0}^{n_i}\frac{(-1)^{k_i}}{k_i!}\binom{n_i}{k_i}\gamma \left ( k_i+1,\frac{\gamma_i}{\overline{\gamma}_i} \right ),
\end{align}
\end{subequations}
where  $i \in \left \{ \mathrm{B},\mathrm{E} \right \}$ %{\color{blue}Pode colocar todas as ocorrencias de B e E no mesmo padrao com mathrm}
represents either the main channel or the eavesdropper channel, $\overline{\gamma}_i$ is the average receive SNR at $\mathrm{B}$ or $\mathrm{E}$ as previously stated, i.e.,
\begin{align}\label{eq:4}
\overline{\gamma}_i=\gamma_0\mathbb{E}\left [R_i^2\right ]r_i^{-\eta_i}=\gamma_0\left ( \sum_{n=0}^{\mathbb{N}_i}V_{n,i}^2+\Omega_i\right )r_i^{-\eta_i},
\end{align}
where $\eta_i$ is the path-loss exponent, and $r_i$ is the propagation distance\footnote{Here, as in the LOS ball blockage model, we assume that $r_i$ lies within a sphere of fixed radius $R_\mathrm{B}$. Interested readers can revise~\cite{Bai} for further guidance about
 simplification of the LOS region as a fixed equivalent
LOS ball in mm-Wave cellular networks.}. The $C_{n_i}$ coefficient can be calculated recursively by~\cite{Chun}
\begin{align}\label{eq:5}
C_{n_i}=\sum_{k_i=0}^{n_i}\frac{\left ( -\epsilon_i \right )^{k_i}}{k_i!}\binom{n_i}{k_i}u_{N_i+1}^{(2k_i)},
\end{align}
where $\epsilon_i=\left ( \mathbb{E}\left [R_i^2\right ] \right )^{-1}$, and
 \begin{align}\label{eq:6}
 u_j^{2k_i}=\sum_{m=0}^{k_i}\binom{k_i}{m}^2u_{j-1}^{(2m)}v_j^{(2k_i-2m)}, \text{ for  $j=2,\dots,N_i+1$},
\end{align}
where the initial value is $u_1^{2k_i}=v_1^{2k_i}$, and 
 \begin{equation}
       v_j^{2k_i}=\begin{cases}
      V_{j,i}^{2k_i}, \hspace{9.5mm} \text{ for  $j=1\dots N_i$},\\
       \left ( 1 \right )_{k_i} \left ( \Omega_i \right )^{k_i},\text{ for  $j=N_i+1$}.
         \end{cases}
\label{eq:7}
\end{equation}

\section{Secrecy Outage Probability Analysis}
\subsection{Exact SOP Analysis}
We consider a passive eavesdropping scenario, so that Alice has no channel state information (CSI) of Eve's channel. Hence, Alice's only choice is to encode the confidential data
into codewords at a constant rate $R_{\mathrm{S}}$. This can occur in a practical setup where Eve is silent during all transmissions. According to~\cite{Wyner}, the secrecy capacity is
obtained as
\begin{align}\label{eq:8}
C_\mathrm{S}&=\!\text{max}\left \{C_\mathrm{B}-C_\mathrm{E},0  \right \} \nonumber \\
%&=\!\text{max}\left \{\log_2\!\left (1\!+\!\frac{|h_{\mathrm{AB}}|^2P_{\mathrm{A}}}{N_{\mathrm{0}}} \!\right )\!-\!\log_2\!\left (1\!+\!\frac{|h_{\mathrm{AE}}|^2P_{\mathrm{A}}}{N_{\mathrm{0}}}\!\right ),0  \right \} \nonumber \\ 
&=\!\text{max}\left \{\log_2(1+\gamma_\mathrm{B})-\log_2(1+\gamma_\mathrm{E}),0  \right \} 
 %       &=\left\{ 
 %       	\begin{array}{ll}
 %       		\hspace*{1mm} \log_2\left ( \frac{1+\gamma_\mathrm{B}}{1+\gamma_\mathrm{E}} \right ), \quad \text{if} \enspace \gamma_\mathrm{B}>\gamma_\mathrm{E}\\
 %       		\hspace*{1mm} 0, \hspace{6em} \text{if} \enspace \gamma_\mathrm{B} \leq \gamma_\mathrm{E},
 %       	\end{array}
 %       \right. \vspace{2mm} 
\end{align}
%{\color{red}It is better to call PS as PA, which is the transmit power at Alice}
%where $P_{\mathrm{A}}$ is the transmit power at Alice, and  $C_i$ and $h_{Ai}$, with $i$ $\in \left \{\mathrm{B},\mathrm{E} \right \}$, are the capacity and the channel coefficient of the main channel and the wiretap channel, respectively. 
With these considerations, secrecy is achieved only in those instants on which
 $R_\mathrm{S}\leq C_{\mathrm{S}}$, and is compromised otherwise. In this context, the SOP is defined as the probability that the instantaneous secrecy capacity $C_\mathrm{S}$ falls below a target secrecy rate threshold $R_{\mathrm{S}}$, and can be expressed as~\cite{Barros}
 \begin{align}\label{eq:sop}
 \text{SOP}&=\Pr\left \{ C_\mathrm{S}\left ( \gamma_\mathrm{B},\gamma_\mathrm{E} \right ) < R_{\mathrm{S}}  \right \}\nonumber \\ 
&=\Pr\left \{ \left ( \frac{1+\gamma_\mathrm{B}}{1+\gamma_\mathrm{E}} \right )< 2^{R_{\mathrm{S}}}\buildrel \Delta \over  = \tau \right \}\nonumber \\ 
 &=\Pr\left \{ \gamma_\mathrm{B}< \tau \gamma_\mathrm{E}+\tau-1 \right \} \nonumber \\ 
 &=\int_{0}^{\infty}F_{\gamma_\mathrm{B}}\left ( \tau \gamma_\mathrm{E}+\tau-1 \right )f_{\gamma_\mathrm{E}}(\gamma_\mathrm{E})d\gamma_\mathrm{E}.
\end{align}
Furthermore, a high SNR approximation of the SOP can be obtained from~\eqref{eq:sop} as
 \begin{align}\label{eq:sopL}
 \text{SOP}_{\text{A}}&=\Pr\left \{ \gamma_\mathrm{B}< \tau \gamma_\mathrm{E}\right \} \leq \text{SOP} .
\end{align}
Substituting~\eqref{eq:SNR} into~\eqref{eq:sop} and~\eqref{eq:sopL}, we can obtain the SOP and the $\text{SOP}_{\text{A}}$, respectively, over NWDP fading channels in the following Lemma.
\begin{lemma}\label{Lema1}
 The SOP and the $\text{SOP}_{\text{A}}$ over NWDP fading channels can be
obtained as~\eqref{eq:SOPExact} and~\eqref{eq:SOPLower}, respectively, at the top of the next page.
%\textcolor{blue}{(Que ventaja tiene SOPa respecto a SOP? Lo digo porque ambas parecen igual de complicadas de calcular a simple vista -- cuando imagino que no es el caso.)}
\end{lemma}
\begin{proof}
	See Appendix~\ref{ap:SOPs}.
\end{proof}

\textbf{Remark 1.} Notice that the derived analytical expressions for both the SOP and $\text{SOP}_{\text{A}}$ are expressed in terms of infinite series representations. This is also the case of the analysis in \cite{TWDP} for TWDP based on the \textit{approximate} PDF in \cite{Durgin}, which arises as a special case of our analysis. 
\subsection{Asymptotic Analysis
}
To get further insights about the role of the fading parameters in the system performance, the main concern in
this section is to derive asymptotic closed-form expressions to investigate the behavior 
of the SOP given in~\eqref{eq:sop} at high-SNR regime. Here, we assume the following scenarios: $(i)$ both $\overline{\gamma}_\mathrm{B}$ and $\overline{\gamma}_\mathrm{E} $ go to infinity,
while the ratio between these SNRs is kept unchanged\footnote{This scenario corresponds to the case when both $\mathrm{B}$ and $\mathrm{E}$ are close to $\mathrm{A}$.}; $(ii)$ $\overline{\gamma}_\mathrm{B}\rightarrow \infty$ while $\overline{\gamma}_\mathrm{E}$ is
kept fixed\footnote{
This scenario corresponds to the case that $\mathrm{A}$ is very close to $\mathrm{B}$ and $\mathrm{E}$ is located far way.}. Our goal would be obtaining asymptotic expressions in the form $\mathrm{SOP}^{\infty}\approx \mathrm{G}_c\overline{\gamma}_\mathrm{B}^{-\mathrm{G}_d}$, where $\mathrm{G}_c$ and $\mathrm{G}_d$ represent the secrecy array gain and the secrecy diversity gain, respectively. However, as we will later see, such expressions will not be of much practical use as in many cases, such asymptotic does not kick in until very low probabilities are considered. Hence, our asymptotic expressions will incorporate additional terms on which the exponent of $\overline\gamma_\mathrm{B}$ play a relevant role on the SOP decay for practical operational values. Next, the corresponding asymptotic expressions of the SOP over NWDP fading channels are given in the following Lemma.
\begin{lemma}\label{Lema2}
The asymptotic closed-form expressions of the SOP over NWDP fading channels for the cases in that both $\overline{\gamma}_\mathrm{B}\rightarrow \infty,\overline{\gamma}_\mathrm{E}\rightarrow \infty$, and only $\overline{\gamma}_\mathrm{B}\rightarrow \infty$ can be
obtained as~\eqref{eq:SOPAsin1} and~\eqref{eq:SOPAsin2}, respectively, at the top of the next page.
\end{lemma}
\begin{proof}
See Appendix~\ref{ap:SOPAsintotas}.
\end{proof}

%%%%%%%%%%%%%%%%%%%%%%%%%%%%%%%%%%%%%%%%%%%%%%%%%%%%%%%  SOP EXACT
\begin{figure*}[ht]
	%\hrulefill
	\begin{normalsize}
\begin{align}\label{eq:SOPExact}
\text{SOP}
 = & 1-\sum_{n_\mathrm{B}=0}^{\infty}C_{n_{\mathrm{B}} }\sum_{k_\mathrm{B}=0}^{n_\mathrm{B}}(-1)^{k_\mathrm{B}}\binom{n_\mathrm{B}}{k_\mathrm{B}}\left ( \frac{1}{\overline{\gamma}_\mathrm{E}}  \right )\sum_{n_\mathrm{E}=0}^{\infty}C_{n_\mathrm{E}}\sum_{q=0}^{k_\mathrm{B}}\frac{1}{q!}\left ( \frac{1}{\overline{\gamma}_\mathrm{B}} \right )^q \exp\left ( -\frac{\tau-1}{\overline{\gamma}_\mathrm{B}} \right ) \sum_{a=0}^{q}\binom{q}{a}\left ( \tau-1 \right )^{q-a}\tau^a  \nonumber  \\ & \times  \left ( \frac{1}{\overline{\gamma}_\mathrm{E}}+\frac{\tau}{\overline{\gamma}_\mathrm{B}} \right )^{-1-a}\Gamma\left ( 1+a \right ){ }_2F_1\left ( 1+a,-n_\mathrm{E};1;\frac{\overline{\gamma}_\mathrm{B}}{\overline{\gamma}_\mathrm{B}+\overline{\gamma}_\mathrm{E} \tau} \right )
 \end{align}
	\end{normalsize}
%	\hrulefill
	%\vspace{-5mm}
\end{figure*}
%%%%%%%%%%%%%%%%%%%%%%%%%%%%%%%%%%%%%%%%%%%%%%%%%%  SOP LOWER
\begin{figure*}[ht!]
	%\hrulefill
	\begin{normalsize}
\begin{align}\label{eq:SOPLower}
\text{SOP}_{\text{A}} 
 = & 1-\sum_{n_\mathrm{B}=0}^{\infty}C_{n_\mathrm{B}}\sum_{k_\mathrm{B}=0}^{n_\mathrm{B}}(-1)^{k_\mathrm{B}}\binom{n_\mathrm{B}}{k_\mathrm{B}}\left ( \frac{1}{\overline{\gamma}_\mathrm{E}} \right ) \sum_{n_\mathrm{E}=0}^{\infty}C_{n_\mathrm{E}}       \sum_{q=0}^{k_\mathrm{B}}\frac{1}{q!}\left ( \frac{\tau}{\overline{\gamma}_\mathrm{B}} \right )^q \left ( \frac{1}{\overline{\gamma}_\mathrm{E}}+\frac{\tau}{\overline{\gamma}_\mathrm{B}} \right )^{-1-q} \Gamma\left ( 1+q \right ) \nonumber\\ &\times 
   { }_2F_1\left (-n_\mathrm{E},1+q;1;\frac{\overline{\gamma}_\mathrm{B}}{\overline{\gamma}_\mathrm{B}+\overline{\gamma}_\mathrm{E} \tau} \right )
 \end{align}
	\end{normalsize}
	%\hrulefill
%	\vspace{-5mm}
\end{figure*}

%%%%%%%%%%%%%%%%%%%%%%%%%%%%%%%%%%%%%%%%%%%%%%%%%%  Asintota 1
\begin{figure*}[ht!]
	%\hrulefill
	\begin{normalsize}
\begin{align}\label{eq:SOPAsin1}
\mathrm{SOP}^\infty \approx 
  & \sum_{n_\mathrm{B}=0}^{\infty}C_{n_\mathrm{B}}\sum_{k_\mathrm{B}=0}^{n_\mathrm{B}}\frac{(-1)^{k_\mathrm{B}} }{\left ( k_\mathrm{B}+1 \right )! }\binom{n_\mathrm{B}}{k_\mathrm{B}}\left ( \frac{\tau}{\overline{\gamma}_\mathrm{B}} \right )^{k_\mathrm{B}+1}\sum_{n_\mathrm{E}=0}^{\infty}C_{n_\mathrm{E}} 
  \sum_{k_\mathrm{E}=0}^{n_\mathrm{E}}\frac{(-1)^{k_\mathrm{E}} }{ k_\mathrm{E}! }\binom{n_\mathrm{E}}{k_\mathrm{E}}\left ( \frac{1}{\overline{\gamma}_\mathrm{E}} \right )^{k_\mathrm{E}+1} \sum_{i=1}^{n}w_i\exp\left ( l_i \right ) l_i^{k_\mathrm{B}+k_\mathrm{E}+1}
 \end{align}
	\end{normalsize}
	%\hrulefill
%	\vspace{-5mm}
\end{figure*}

%%%%%%%%%%%%%%%%%%%%%%%%%%%%%%%%%%%%%%%%%%%%%%%%%%  Asintota 2
\begin{figure*}[ht!]
	%\hrulefill
	\begin{normalsize}
\begin{align}\label{eq:SOPAsin2}
\mathrm{SOP}^{\infty}\approx 
  & \sum_{n_\mathrm{B}=0}^{\infty}C_{n_\mathrm{B}}\sum_{k_\mathrm{B}=0}^{n_\mathrm{B}}\frac{(-1)^{k_\mathrm{B}} }{\left ( k_\mathrm{B}+1 \right )! }\binom{n_\mathrm{B}}{k_\mathrm{B}}\left ( \frac{\tau\overline{\gamma}_\mathrm{E}}{\overline{\gamma}_\mathrm{B}} \right )^{k_\mathrm{B}+1}\sum_{n_\mathrm{E}=0}^{\infty}C_{n_\mathrm{E}} \frac{\Gamma\left ( k_\mathrm{B}+2 \right )\Gamma\left ( n_\mathrm{E}+1 \right )}{n_\mathrm{E}!}{ }_2F_1\left (-n_\mathrm{E},k_\mathrm{B}+2 ;1;\overline{\gamma}_\mathrm{E} \right ) 
 \end{align}
	\end{normalsize}
	\hrulefill
%	\vspace{-5mm}
\end{figure*}

%%%%%%%%%%%%%%%%%%%%%%%%%%%%%%%%
%NUMERICAL RESULTS
\section{Numerical results and discussions} \label{sect:numericals}
In this section, we validate the accuracy of the proposed expressions{\footnote{Here, it is worth mentioning that depending on the value of the involved channel parameters, these series require different number of terms to attain an accurate approximation. In this context, the overall convergence speed of these series is achieved faster for small values of both dominant rays (e.g, $\mathrm{N}_\mathrm{B}$, and $\mathrm{N}_\mathrm{E}$) and power of Bob's dominant specular components (i.e., $K_{\mathrm{dB}}^{\mathrm{B}}$). For instance, exhaustive tests have shown that the number of terms to arrive at the desired accuracy (e.g., $10^{-6}$) varied from 20 to 30 at Bob and from 4 to 10 at Eve, and the average elapsed times to obtain the aforementioned accuracy were $\sim 14.1, 27.5, 81.7, 103.5, 114.6$ seconds for $\mathrm{N}=1,\dots,5$, respectively. Moreover, the mathematical representation of the derived series consists of well-known elementary and special functions, which can be easily implemented in software for numerical evaluation.}}. for some representative cases via Monte Carlo simulations. %Without loss of generality, 
%we normalize the propagation distance $r_i$, with $i \in \left \{\mathrm{B},\mathrm{E}.\right \}$, to unity. 
We define a power ratio parameter similar to the well-known Rician $K$ parameter, e.g., $K_{\mathrm{N}_i}\buildrel \Delta \over = \frac{\Omega_{\mathrm{N}_i}}{\Omega_i}$, with $\Omega_{\mathrm{N}_i}=\sum_{n=0}^{\mathrm{N}_i}V_{n,i}^2$ being the total average power of the specular components. %Also, in all the figures within this section, we set the average power of the diffuse component with unit-power i.e, $\Omega_\mathrm{B}=\Omega_\mathrm{E}=1$ and 
For the sake of comparison, the Rayleigh case (i.e., $\mathrm{N}_\mathrm{B}=\mathrm{N}_\mathrm{E}=0$) is included as a reference. 

Before getting into the numerical examples, an important remark is in order. Herein, we
emphasize on providing clear evidence to identify the impact of increasing/decreasing both the number and the power of the dominant specular waves over the secrecy performance. In other words, we aim to determine to what extent it is worth that each of the individual specular waves is treated separately, or it can be safely incorporated into the diffuse component.
% figure  1 sop
\begin{figure}[t]
\centering 
\psfrag{H}[Bc][Bc][0.6]{$K_{\mathrm{dB}}^{\mathrm{B}}=25$ dB}
\psfrag{Z}[Bc][Bc][0.6][0]{$K_{\mathrm{dB}}^{\mathrm{B}}=15$ dB}
\includegraphics[width=1\linewidth]{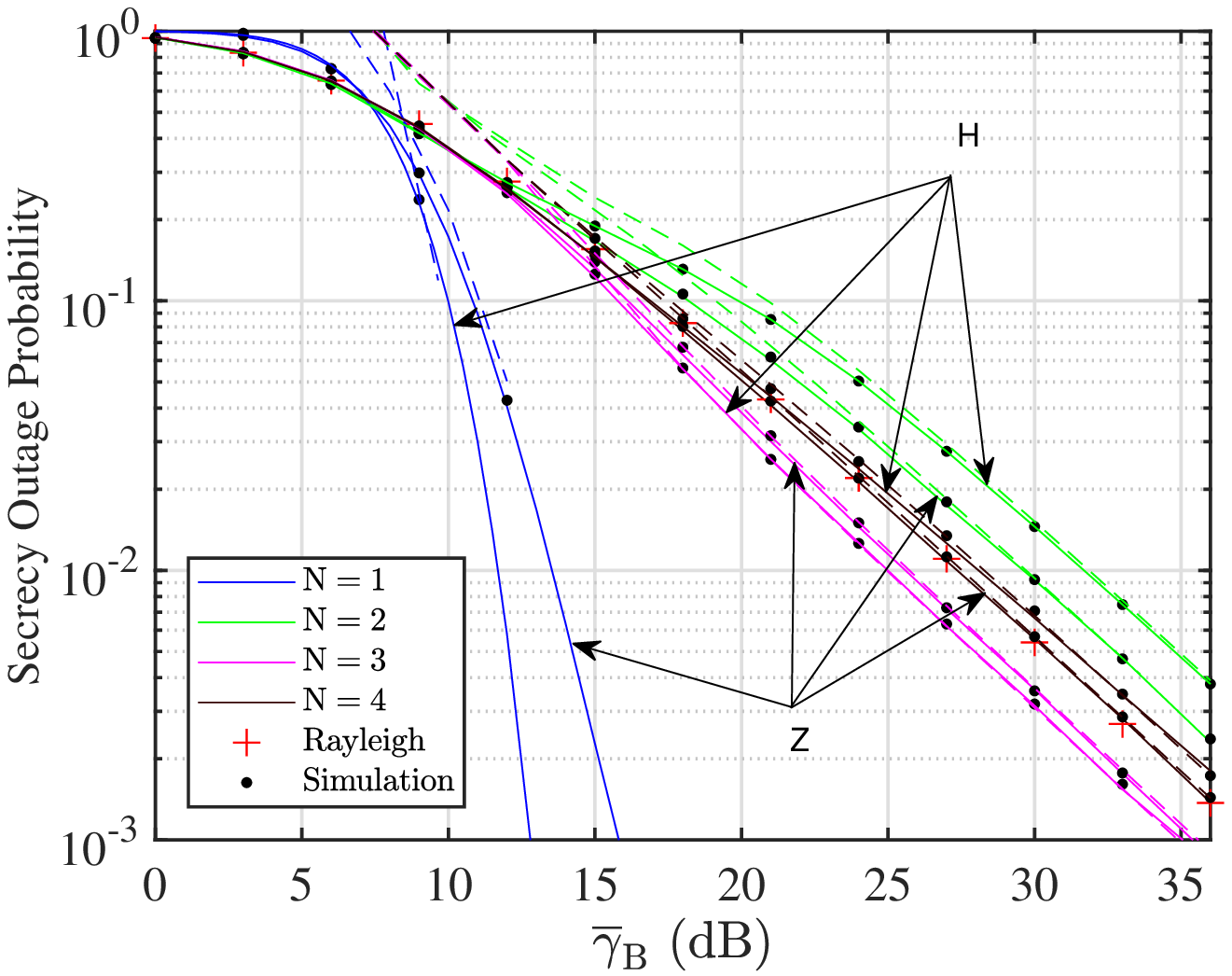} \caption{SOP vs. $\overline{\gamma}_\mathrm{B}$, for different numbers of dominant
specular waves $\mathrm{N}$ by considering balanced amplitude scenario (i.e., $V_{n,i}=1$ for $n=1,\dots, \mathrm{N}_i$). For all curves, the parameter values are:   $R_{\mathrm{S}}$ = 1 bps/Hz, $\overline{\gamma}_\mathrm{E}=4$ dB, $K_{\mathrm{dB}}^{\mathrm{E}}=10$ dB, 
$\Omega_i=1$, and $\mathrm{N}_i=\mathrm{N}$ for $i \in \left \{ \mathrm{B},\mathrm{E} \right \}$. Dashed lines  correspond to asymptotic analysis by using expression~\eqref{eq:SOPAsin1}.}
\label{fig1Sop}
\end{figure}

In Fig.~\ref{fig1Sop}, we compare the SOP as a function of $\overline{\gamma}_\mathrm{B}$ for
different values of dominant specular components $\mathrm{N}_\mathrm{B}=\mathrm{N}_\mathrm{E}=\mathrm{N}$, by considering the case of balanced amplitudes, i.e., $V_{n,\mathrm{B}}=V_{n,\mathrm{E}}$  $\forall \mathrm{B}$, $n=1 \dots \mathrm{N}$. For this scenario, the corresponding fading parameters are given by: $K_{\mathrm{dB}}^{\mathrm{B}}=10 \log_{10} \left ( K_{\mathrm{B}}=K_{\mathrm{N}_\mathrm{B}} \right )\in \left \{ 8,25 \right \}$ dBs with $K_{\mathrm{dB}}^{\mathrm{E}}=0$ dB, $R_{\mathrm{S}}$ = 1 bps/Hz, and $\overline{\gamma}_\mathrm{E}=1$ dB. Note that in all instances, Monte Carlo simulations perfectly match with our derived results. 

We see that the secrecy performance does not monotonically increase with the number of specular components; instead, we see that the cases with $\mathrm{N}_\mathrm{B}=1$ and $\mathrm{N}_\mathrm{B}=2$ bound the secrecy performance when the rest of parameters are fixed. This is in coherence with the fact that for an even number of dominant specular components of equal amplitudes, the probability of total cancellation between them is larger than when an odd number is considered \cite{Abdi}. This increases fading severity for the desired link more relevantly for large $K$, which ultimately degrades the SOP. We also see that for $\mathrm{N}=4$, the performance is very similar than in the Rayleigh case (i.e., $\mathrm{N}\rightarrow\infty$).

%Observe also that the SOP performance improves with increasing the power of Bob's dominant specular components (i.e., $K_{\mathrm{dB}}^{\mathrm{B}}$) only for odd numbers of rays (i.e, $N=1, 3$). In the opposite scenario, i.e., when increasing $K_{\mathrm{dB}}^{\mathrm{B}}$ for even numbers of dominant specular waves (i.e, $N=2, 4$), we obtain a loss in secrecy performance.

%\textcolor{blue}{(Por el valor que has escogido de la K de Eve, el efecto de aumentar el numero de rayos en Eve queda enmascarado. Para el canal del fisgon pasa algo similar: con numero par de rayos, es mas probable que se cancelen y por tanto el fading es mas severo en ese enlace. Que el fading del fisgon sea severo es bueno para nosotros, por lo q la SOP mejoraria en ese caso: solo que no se si seraa tan notoria la mejora al ser la SNR del fisgon baja.)} 

%Therefore, it is clear that in the case of balanced amplitudes the impact on the SOP performance by having even numbers of dominants waves in both Bob and Eve is detrimental as $K_{\mathrm{dB}}^{\mathrm{B}}$ increases. Also,  we see that the best performances of the SOP are attained for the Rician cases (i.e.,  $N = 1$), and the Rayleigh's case behaves similarly that $N=2,3,4$ with $K_{\mathrm{dB}}^{\mathrm{B}}=8$ dB.

In Fig.~\ref{fig2Sop}, we now evaluate the SOP vs. $\overline{\gamma}_\mathrm{B}$ for
different numbers of dominant specular components $\mathrm{N}_\mathrm{B}=\mathrm{N}_\mathrm{E}=\mathrm{N}$ by considering an unbalanced amplitudes scenario. For simplicity, yet without loss of
generality, the amplitudes of successive rays are expressed in terms of the amplitude of the first dominant component, as proposed in~\cite{Espinosa}, that is, $V_{n,i}=\alpha_{n,i} V_{1,i}$ for $n=2, \dots, \mathrm{N}_i$, with $0<\alpha_{n,i}<1$ and $i$ $\in \left \{\mathrm{B},\mathrm{E} \right \}$. Considering this, we set: $\alpha_{n,i}=\alpha_\mathrm{B}=\alpha_\mathrm{E}=0.3$ with $K_{\mathrm{dB}}^{\mathrm{B}}\in\left \{8, 25 \right \}$ dBs, $K_{\mathrm{dB}}^{\mathrm{E}}=0$ dB, $R_{\mathrm{S}}$ = 1 bps/Hz, and $\overline{\gamma}_\mathrm{E}=1$ dB. Here, we investigate the impact of increasing both the number and the power of Bob's dominant rays for the case of unbalanced amplitudes. We observe in all traces that, unlike on the balanced counterpart, the SOP performance now monotonically improves when rising $K_{\mathrm{dB}}^{\mathrm{B}}$ or lowering $\mathrm{N}$, regardless of whether it is even or odd. It can be observed that a reduced number of dominant specular components at Bob is now beneficial from a secrecy perspective. We also see that in all cases, the SOP performance is always better than its Rayleigh counterpart. Regarding the asymptotic behaviour, it can
be noticed that the asymptotic plots accurately approximate the SOP curves at high SNR regime. Additionally, all curves have different slopes. The reason for this will be discussed later.

% figure  2 sop
\begin{figure}[t]
\centering 
\psfrag{H}[Bc][Bc][0.6]{$K_{\mathrm{dB}}^{\mathrm{B}}=25$ dB}
\psfrag{Z}[Bc][Bc][0.6][0]{$K_{\mathrm{dB}}^{\mathrm{B}}=8$ dB}
\includegraphics[width=1\linewidth]{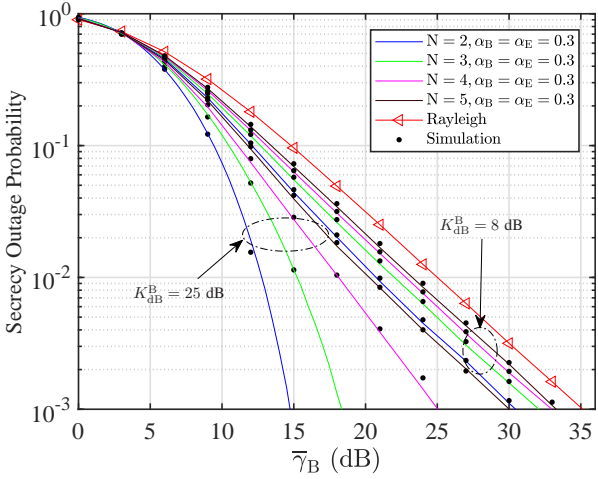} \caption{SOP in terms of $\overline{\gamma}_\mathrm{B}$ for different numbers of dominant
specular waves $\mathrm{N}$, by considering unbalanced amplitude case (i.e., $\alpha_{n,i}=\alpha_i=0.3$). For all cases, the corresponding parameters are set to the following values: $R_{\mathrm{S}}$ = 1 bps/Hz, $\overline{\gamma}_E=1$ dB, $K_{\mathrm{dB}}^{\mathrm{E}}=0$ dB, $\Omega_i=1$, and $\mathrm{N}_i=\mathrm{N}$ for $i \in \left \{ \mathrm{B},\mathrm{E} \right \}$.}
\label{fig2Sop}
\end{figure}

In Fig.~\ref{fig3Sop}, we evaluate both the SOP and the $\text{SOP}_{\text{A}}$ as a function of  $\overline{\gamma}_\mathrm{B}$, in order to understand the interplay between the  number of dominant specular components $\mathrm{N}_\mathrm{B}=\mathrm{N}_\mathrm{E}=\mathrm{N}$, the amplitude imbalance and the power of the dominant components. We use a similar set of parameters as those in Fig.~\ref{fig2Sop}, except for $\left \{ \alpha_\mathrm{B}, \alpha_\mathrm{E} \right \}=\left \{ 0.2, 0.9 \right \} $, $\left \{ \alpha_\mathrm{B}, \alpha_\mathrm{E} \right \}=\left \{ 0.9, 0.2 \right \} $, and $\overline{\gamma}_\mathrm{E}=8$ dB. %\footnote{Here, unlike the cases described previously, we set $\overline{\gamma}_\mathrm{E}=8$ in order to verify the accuracy of the high SNR approximation given in~\eqref{eq:SOPLower}.}. 
We now observe that the worst secrecy performance is attained for cases where the imbalance for the legitimate user $\alpha_\mathrm{B}$ is smaller than that of $\alpha_\mathrm{E}$, i.e.  when $\alpha_\mathrm{B}>\alpha_\mathrm{E}$. Therefore, for the cases where $\alpha_\mathrm{B}$ is lower than $\alpha_\mathrm{E}$, we can obtain the desired secrecy performance (i.e., a target SOP) for a lower average SNR at Bob. In such case, some other interesting observations can be made: $(i)$ the SOP performance under NWDP fading is much better than in the Rayleigh case, and $(ii)$ the increase on the number of dominant specular rays arriving at Bob is detrimental from a secrecy perspective. 

On the other hand, the worst SOP performances are achieved for the case with ($\alpha_\mathrm{B}=0.9, \alpha_\mathrm{E}=0.2$), which is explained as follows: because the amplitudes for the legitimate link are balanced, this is translated into a more severe fading; conversely, the unbalanced amplitudes for the eavesdropper's link indicate a lower fading severity compared to the Rayleigh case. Combining the two effects, the overall SOP performance is hence worse than when assuming Rayleigh fading for both links.

%Regarding to the high SNR approximation of the SOP, it is clear that its performance is sufficiently tight with regard the exact analytical solution. It is also expected that the $\text{SOP}_{\text{A}}$ gradually approximates the exact SOP with higher accuracy as $\overline{\gamma}_\mathrm{E}$ increases. 

%, despite of a decreasing average SNR at Bob $\overline{\gamma}_\mathrm{B}$. Moreover, note that in such a scenario (i.e., $\alpha_\mathrm{B}=0.2, \alpha_\mathrm{E}=0.9$), we have that:

%{\color{red}Regarding to the lower bound of the SOP, it is clear that its performance is not very accurate with regard the analytical solution. The reason for this is because, in all curves, we set $\overline{\gamma}_\mathrm{E}=0$ dB. However, it is expected that high-SNRs of $\overline{\gamma}_\mathrm{E}$ make the lower bound of the SOP sufficiently tight with the exact SOP. }
% figure  3 sop
\begin{figure}[t]
\centering 
\psfrag{H}[Bc][Bc][0.6]{$\alpha_{\mathrm{B}}=0.2$, $\alpha_{\mathrm{E}}=0.9$}
\psfrag{Z}[Bc][Bc][0.6][0]{$\alpha_{\mathrm{B}}=0.9$, $\alpha_{\mathrm{E}}=0.2$}
\includegraphics[width=1\linewidth]{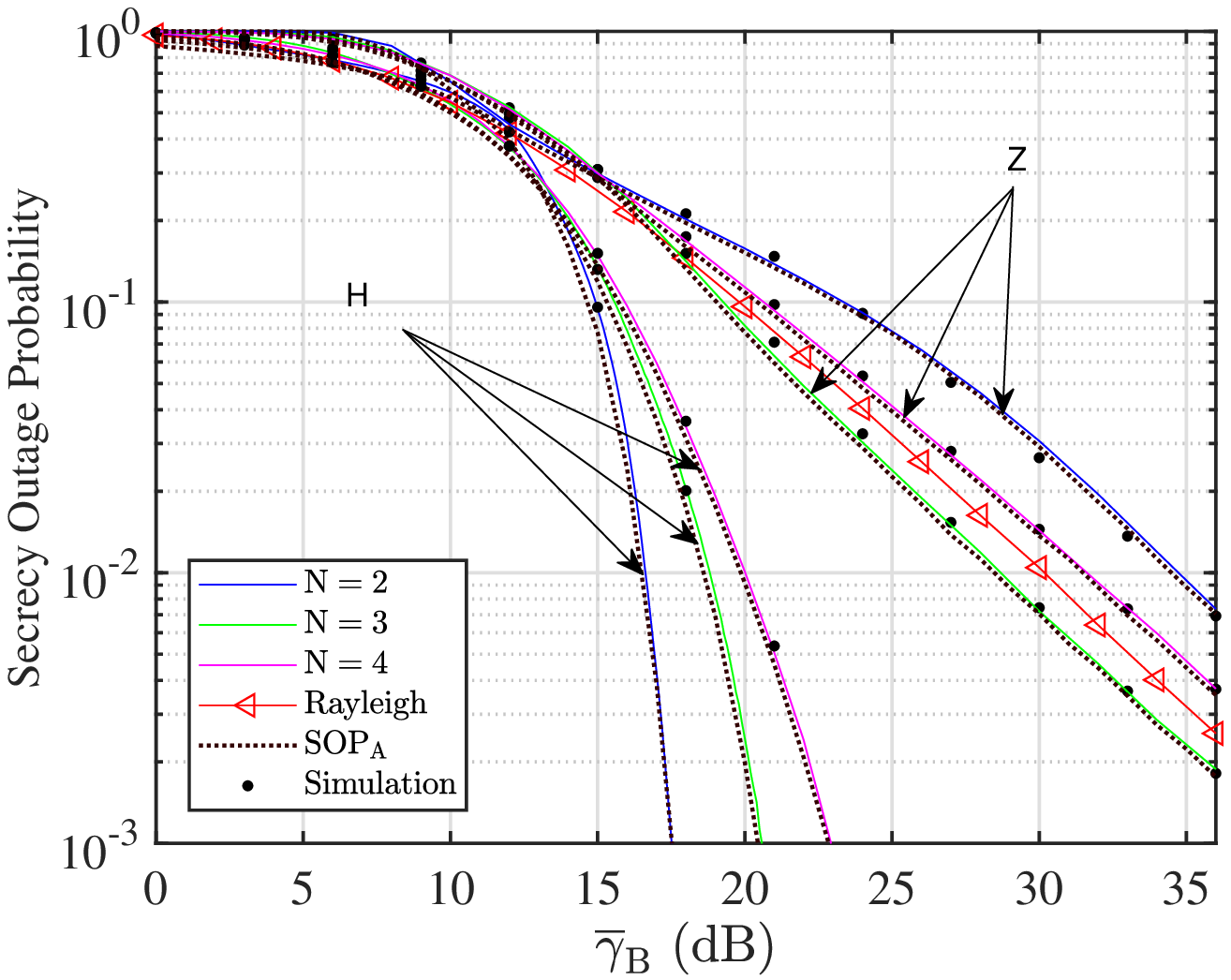} \caption{SOP vs. $\overline{\gamma}_\mathrm{B}$ for different numbers of dominant
specular waves $\mathrm{N}$ by considering unbalanced amplitude case (i.e., $\left \{ \alpha_\mathrm{B}, \alpha_\mathrm{E} \right \}=\left \{ 0.2, 0.9 \right \} $ and $\left \{ \alpha_\mathrm{B}, \alpha_\mathrm{E} \right \}=\left \{ 0.9, 0.2 \right \} $). For all curves, the values of channel parameters are: $R_{\mathrm{S}}$ = 1 bps/Hz, $\overline{\gamma}_\mathrm{E}=8$ dB, $K_{\mathrm{dB}}^{\mathrm{B}}=K_{\mathrm{dB}}^{\mathrm{E}}=25$ dB $\Omega_i=1$, and $\mathrm{N}_i=\mathrm{N}$ for $i \in \left \{ \mathrm{B},\mathrm{E} \right \}$   }
\label{fig3Sop}
\end{figure}

Fig.~\ref{fig4Sop} presents the evolution of the SOP as a function of $R_{S}$, considering the following channel settings: $\overline{\gamma}_\mathrm{E}=1$ dB, $\overline{\gamma}_\mathrm{B}=8$ dB, $K_{\mathrm{dB}}^{\mathrm{B}}=K_{\mathrm{dB}}^{\mathrm{E}}=20$ dB, and $\left \{ \alpha_\mathrm{B}, \alpha_\mathrm{E} \right \}$$=$$\left \{ 0.2, 0.3 \right \}$. Herein, we analyze the effect of having a different number of dominant specular rays at both Bob and Eve over the secrecy performance. We consider the cases $\mathrm{N}_\mathrm{E}=\left \{2,3 \right \}$ and $\mathrm{N}_\mathrm{B}=\left \{2,3,4,5 \right \}$, and for the sake of comparison, we also include the case $\mathrm{N}_\mathrm{B} = \mathrm{N}_\mathrm{E}$. Once again we see that having a larger number of multipath waves at the legitimate receiver in this unbalanced scenario effectively increases channel variability, which causes the SOP obtained when transmitting at a rate $R_\mathrm{S}$ to be increased with $\mathrm{N}_\mathrm{B}$. We also see that for a fixed $\mathrm{N}_\mathrm{B}$, increasing the number of rays on the eavesdropper's channel is also detrimental for the SOP. This can be understood by recalling that in the presence of a single dominant specular component for each link and a strong LOS condition, the set-up almost reduces to the AWGN case, for which the SOP is zero as $\overline{\gamma}_\mathrm{B}>\overline{\gamma}_\mathrm{E}$. Hence, having a reduced number of rays and a dominant component much larger than the remaining specular waves turn out being beneficial for physical layer security.

%{\color{blue}Esta ultima afirmacao precisa de mais discussoes, por exemplo, faz sentido piorar o secrecy se tem mais caminhos legitimos que o eavesdropper possa interceptar. REcomendo trabalhar melhor e tentar explicar esta observacao}
% figure   sop vs Rs
\begin{figure}[t]
\centering 
\psfrag{Z}[Bc][Bc][0.6][0]{$K_{\mathrm{dB}}^{\mathrm{B}}=K_{\mathrm{dB}}^{\mathrm{E}}=20$ dB}
\includegraphics[width=1\linewidth]{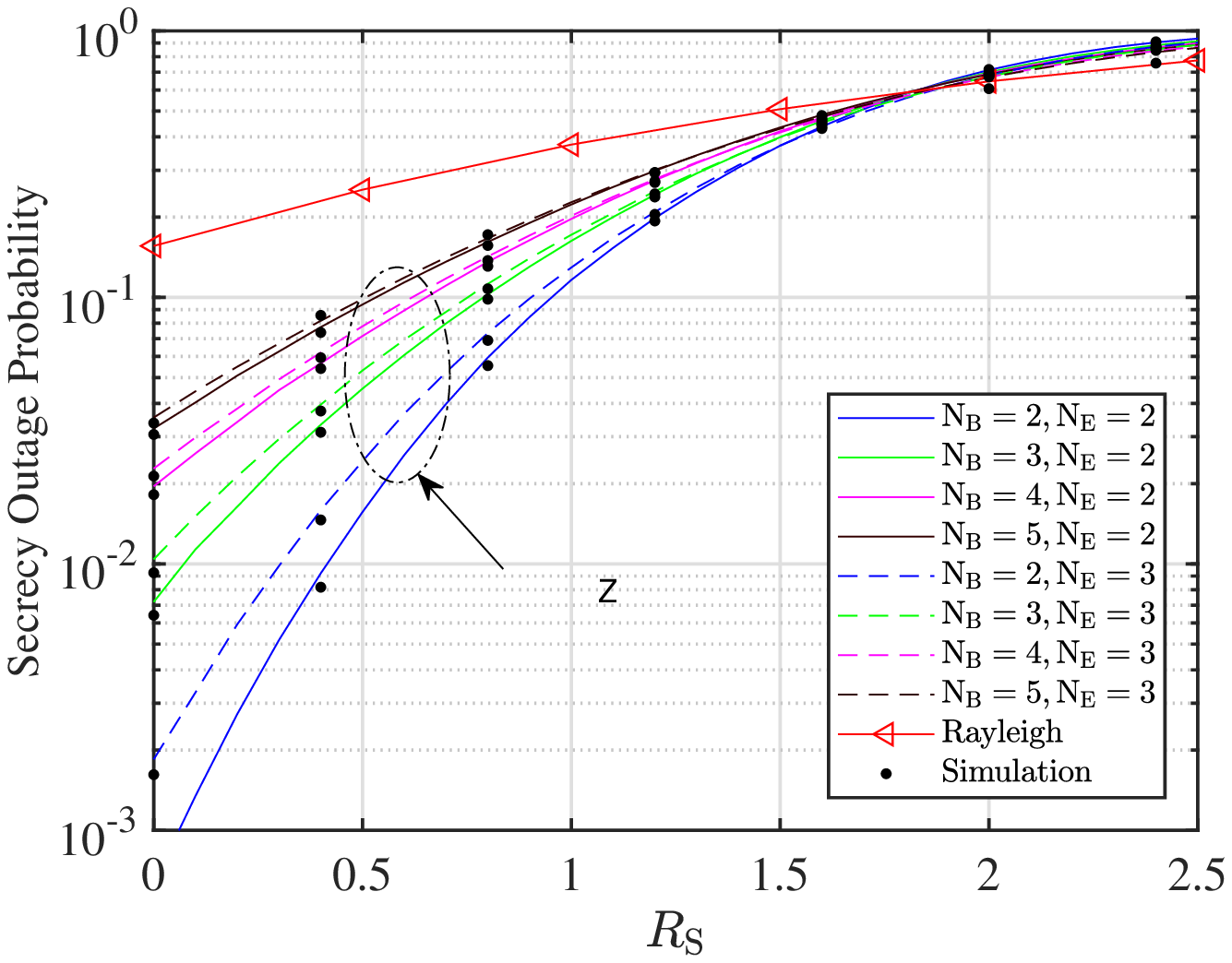} \caption{SOP in terms of $R_{\mathrm{S}}$ for different numbers of dominant
specular waves of $\mathrm{N}_\mathrm{B}=\left \{2,3,4,5 \right \}$  with
regard to $\mathrm{N}_\mathrm{E}=\left \{2,3 \right \}$ by considering unbalanced amplitude case (i.e., $\left \{ \alpha_\mathrm{B}, \alpha_\mathrm{E} \right \}=\left \{ 0.2, 0.3 \right \} $). For all curves, we set: $\overline{\gamma}_\mathrm{E}=1$ dB, $\overline{\gamma}_\mathrm{B}=8$ dB, $K_{\mathrm{dB}}^{\mathrm{B}}=K_{\mathrm{dB}}^{\mathrm{E}}=20$ dB, and $\Omega_i=1$ for $i \in \left \{ \mathrm{B}, \mathrm{E} \right \}$.}
\label{fig4Sop}
\end{figure}

% figure 6 sop
\begin{figure}[t]
\centering 
\psfrag{Z}[Bc][Bc][0.6][0]{$R_{\mathrm{S}}=2.5$}
\psfrag{V}[Bc][Bc][0.6][0]{$R_{\mathrm{S}}=1$}
\includegraphics[width=1\linewidth]{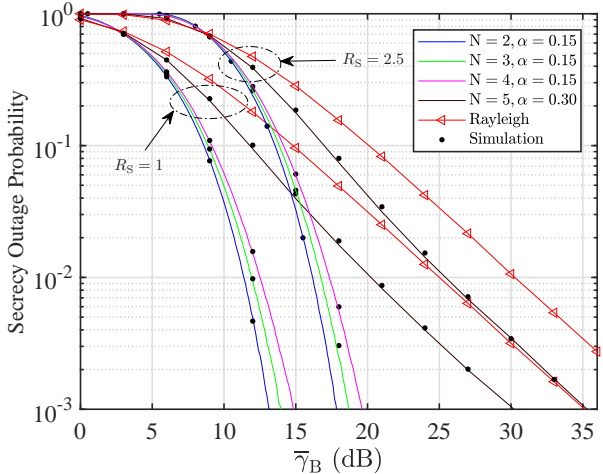} \caption{SOP vs. $\overline{\gamma}_\mathrm{B}$ for different numbers of dominant specular waves $N=N_\mathrm{B}=N_\mathrm{E}$ by varying the value
of $R_{\mathrm{S}}$ and assuming unbalanced amplitudes (i.e, $\alpha=\left \{0.15, 0.30 \right \}$). Also, $\overline{\gamma}_\mathrm{E}=1$ dB, $K_{\mathrm{dB}}^{B}=25$ dB, and $K_{\mathrm{dB}}^{\mathrm{E}}=0$ dB. }
\label{fig5Sop}
\end{figure}

Next, Fig.~\ref{fig5Sop} illustrates the SOP vs. $\overline{\gamma}_\mathrm{B}$ for different numbers of rays $\mathrm{N}_\mathrm{B}=\mathrm{N}_\mathrm{E}=\mathrm{N}$ with $\overline{\gamma}_\mathrm{E}=1$ dB, $K_{\mathrm{dB}}^{\mathrm{B}}=25$ dB, and $K_{\mathrm{dB}}^{\mathrm{E}}=0$ dB. Moreover, we set: $R_{\mathrm{S}}=\left \{1, 2.5 \right \}$ with  $\alpha=\alpha_\mathrm{B},=\alpha_\mathrm{E}=0.15$ and $\alpha=\alpha_\mathrm{B},=\alpha_\mathrm{E}=0.30$ for $\mathrm{N}=2, \dots 4$, and $\mathrm{N}=5$, respectively. From the figure can be observed that both the relative amplitudes and the number of the dominant waves play a pivotal role on the secrecy performance. For instance, we see that the decay of the SOP is rather abrupt for $\alpha=0.15$ and $\mathrm{N}=1,\dots,4$ regardless of the choice of $R_{\mathrm{S}}=\left \{1, 2.5 \right \}$. However, when both the number of rays and the relative amplitudes of the rays are slightly increased (say $\alpha=0.3$ and $\mathrm{N}=5$), then the SOP is dramatically impaired and the decay is now similar to the Rayleigh case. This is in coherence with the observations made in~\cite{Espinosa} in the limit case of the absence of diffuse scattering, as $\alpha \left ( \mathrm{N}_\mathrm{B}-1 \right )<1$.  %Again, all curves over NWDP fading outperform the SOP performance to its corresponding Rayleigh fading cases, which demonstrate the usefulness of this channel model on allowing a better understand for an appropriate resource allocation in real systems.

Finally, in Fig.~\ref{fig6Asintota}, we plot  the SOP vs. $\overline{\gamma}_\mathrm{B}$ and the two asymptotic results~\eqref{eq:SOPAsin1},~\eqref{eq:SOPAsin2}, respectively. In  all the cases, we employ  equal numbers of rays at both $\mathrm{B}$ and $\mathrm{E}$, i.e., $\mathrm{N}=\mathrm{N}_\mathrm{B}=\mathrm{N}_\mathrm{E}$, $R_\mathrm{S}=1$ bps/Hz, $\Omega_\mathrm{B}=\Omega_\mathrm{E}=1$, and $\overline{\gamma}_\mathrm{E}=8$ dB. Also, yet without loss of generality, we assume the following cases: \textit{Case $\mathrm{I}$ $(\mathrm{N}=1)$:} Balanced amplitudes, $V_{1,\mathrm{B}}=V_{1,\mathrm{E}}=1$, and $K_{\mathrm{dB}}^{\mathrm{B}}=K_{\mathrm{dB}}^{\mathrm{E}}=10$ dB; \textit{Case $\mathrm{II}$ $(\mathrm{N}=2)$:} Unbalanced amplitudes, $V_{2,i}=\alpha_{2,i} V_{1,i}$ with $V_{1,i}=1$ for $i$ $\in \left \{\mathrm{B},\mathrm{E} \right \}$, $\alpha_{2,\mathrm{B}}=0.2$, $\alpha_{2,\mathrm{E}}=0.9$, and $K_{\mathrm{dB}}^{\mathrm{B}}=K_{\mathrm{dB}}^{\mathrm{E}}=15$ dB; \textit{Case $\mathrm{III}$ $(\mathrm{N}=3)$:} Unbalanced amplitudes, $V_{n,i}=\alpha_{n,i} V_{1,i}$ with $V_{1,i}=1$ for $i$ $\in \left \{\mathrm{B},\mathrm{E} \right \}$, $\alpha_{n,i}=\alpha_\mathrm{B}=\alpha_\mathrm{E}=0.3$ for $n=2,3.$, and $K_{\mathrm{dB}}^{\mathrm{B}}=K_{\mathrm{dB}}^{\mathrm{E}}=10$ dB. Here, our primary aim is to analyze the secrecy diversity
order of the main links in the proposed scenarios. Firstly,
based on the asymptotic expressions (i.e.,~\eqref{eq:SOPAsin1}, and ~\eqref{eq:SOPAsin2}), we see that the exponents for the $\overline\gamma_\mathrm{B}$ terms depend on one of the summation indexes (i.e., $(k_{\mathrm{B}}+1)$). This suggests that each of these terms contributes in different ways to the decay of the SOP, which explains that the slope of the SOP is different depending on the range of values of $\overline\gamma_\mathrm{B}$. As the SNR is increased, it is only the first term of the series which contributes to the SOP, revealing a diversity order of one (vide all cases in Fig.~\ref{fig6Asintota}). However, we can see that such diversity order is not useful for \textit{Case $\mathrm{II}$}, which justifies the need of the more accurate asymptotic expressions here provided, compared to those only relying on $\mathrm{G}_d$ expression. Also, we can observe that the asymptotic analytical in~\eqref{eq:SOPAsin1} is tighter than the asymptotic one given in~\eqref{eq:SOPAsin2}.
% figure asintotas sop
\begin{figure}[t]
\centering 
\psfrag{Z}[Bc][Bc][0.6][0]{\textit{Case $\mathrm{I}$}}
\psfrag{H}[Bc][Bc][0.6][0]{\textit{Case $\mathrm{III}$}}
\psfrag{V}[Bc][Bc][0.6][0]{\textit{Case $\mathrm{II}$}}
\includegraphics[width=1\linewidth]{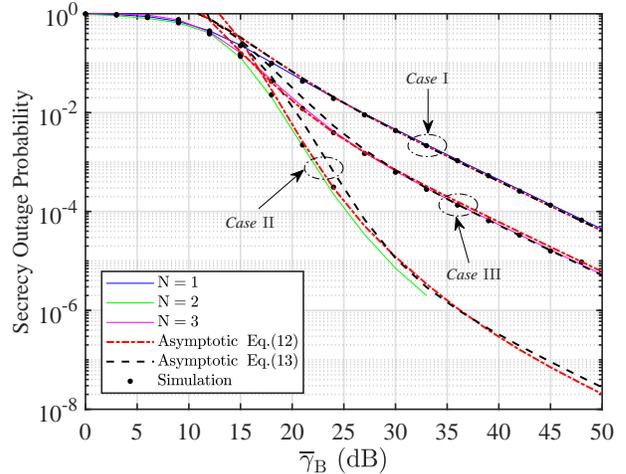} \caption{SOP vs. $\overline{\gamma}_\mathrm{B}$ for different numbers of dominant specular waves $N=N_\mathrm{B}=N_\mathrm{E}$ by  assuming for all cases $R_\mathrm{S}=1$ bps/Hz, $\Omega_\mathrm{B}=\Omega_\mathrm{E}=1$, and $\overline{\gamma}_\mathrm{E}=8$ dB.}
\label{fig6Asintota}
\end{figure}
%%%%%%%%%%%%%%%%%%%%%%%%%%%%%%%%
%CONCLUSIONS

\section{Conclusions}
We investigated how the explicit consideration of the incident waves arriving at the receiver ends may impact physical layer security performance in the context of wireless fading channels. The analytical results here presented complement and generalize those previously reported in the literature, and support the need of using ray-based fading models in those situations on which a reduced number of multipath waves are considered. The main takeaways of our work can be summarized as: $(i)$ abundant
dominant specular rays impair the SOP, so scenarios
with a reduced number of rays arriving at both
Bob and Eve are beneficial whenever $\overline\gamma_B>\overline\gamma_E$; $(ii)$ balanced amplitudes for the eavesdropper's link and unbalanced amplitudes for the desired link are the most favorable case from a PLS perspective; $(iii)$ a significant increase
on the power of Bob's dominant specular components with respect to the power of Eve's dominant specular components (i.e., $K_\mathrm{B}>>K_\mathrm{E}$), in the case of balanced amplitudes, is a worst case scenario for secrecy performance.

%In this work, we have investigated PLS over NWDP fading channels. Particularly, novel exact analytical expressions for the SOP and the lower bound of the SOP in terms of well-known functions were derived. For illustration purposes, numerical results for SOP performance by varying the channel parameters were presented from representative cases. The accuracy of our analytical expressions were efficiently validated by Monte Carlo simulations. Our primary aim was to provide some useful insights for improving the PLS over a more realistic fading channel, in this case we adopt the NWDP fading model. There can highlight the following aspects: $(i)$ plentiful dominant specular rays lead to impair of the SOP, so scenarios where there is a reduced number of rays arriving at both Bob and Eve are beneficial to achieve an improved secrecy performance, $(ii)$ in the case of unbalanced amplitudes, the fact that Eve's relative amplitudes are large enough compared to Bob's relative amplitudes (i.e., $\alpha_\mathrm{E}>\alpha_\mathrm{B}$) makes it suitable to get a better secrecy performance, $(iii)$ a significant increase of the power of Bob's dominant specular components with respect to the power of Eve's dominant specular components (i.e., $K_{\mathrm{dB}}^{\mathrm{B}}>K_{\mathrm{dB}}^{\mathrm{E}}$), for balanced amplitudes scenario, is a worst case scenario from a secrecy perspective. However, this scenario is favorable for unbalanced amplitudes. 

\appendices
%%%%%%%%%%%%%%%%%%%%%%%%%%%%%%%%
%%%%%%%%%%%%%%%%%%%%%%%%%%%%%%%%
%APPENDIX A
%\newpage
\section{ Proofs of Lemma~\ref{Lema1}}
\label{ap:SOPs}
\subsection{$\mathrm{SOP}$}
Substituting~\eqref{eq:SNR} into~\eqref{eq:sop}, we can obtain
\begin{align}\label{eq1ApendiceB}
 \text{SOP}=&\underset{C_1}{\underbrace{\sum_{n_\mathrm{B}=0}^{\infty}C_{n_\mathrm{B}}\sum_{k_\mathrm{B}=0}^{n_\mathrm{B}}\frac{(-1)^{k_\mathrm{B}}}{k_\mathrm{B}!}\binom{n_\mathrm{B}}{k_\mathrm{B}} \frac{1}{\overline{\gamma}_\mathrm{E}} \sum_{n_\mathrm{E}=0}^{\infty} C_{n_\mathrm{E}} }} \nonumber \\ &\times \underset{I_1}{\underbrace{\int_{0}^{\infty} \exp\left ( -\frac{\gamma_\mathrm{E}}{\overline{\gamma}_\mathrm{E}} \right )L_{n_\mathrm{E}}\left ( \frac{\gamma_\mathrm{E}}{\overline{\gamma}_\mathrm{E}} \right )}}  \nonumber \\ &\times \underset{I_1}{\underbrace{\gamma \left ( k_\mathrm{B}+1,\frac{\left ( \tau \gamma_\mathrm{E}+\tau-1 \right )}{\overline{\gamma}_\mathrm{B}} \right )d\gamma_\mathrm{E}.}}
 \end{align}
 Using~\cite[Eq.~(8.352.1)]{Gradshteyn} into~\eqref{eq1ApendiceB}, $I_1$ can be rewritten as
 \begin{align}\label{eq2ApendiceB}
I_1=& \underset{I_2}{\underbrace{ k_\mathrm{B}! \int_{0}^{\infty} \exp\left ( -\frac{\gamma_\mathrm{E}}{\overline{\gamma}_\mathrm{E}} \right )L_{n_\mathrm{E}}\left ( \frac{\gamma_\mathrm{E}}{\overline{\gamma}_\mathrm{E}} \right ) d\gamma_\mathrm{E}}} \nonumber \\ &- k_\mathrm{B}! \sum_{q=0}^{k_\mathrm{B}}\frac{1}{q!}\left ( \frac{1}{\overline{\gamma}_\mathrm{B}} \right )^q \underset{I_3}{\underbrace{\int_{0}^{\infty}\exp\left ( -\frac{\gamma_\mathrm{E}}{\overline{\gamma}_\mathrm{E}} \right ) }}\nonumber \\ &\times  \underset{I_3}{\underbrace{L_{n_\mathrm{E}}\left ( \frac{\gamma_\mathrm{E}}{\overline{\gamma}_\mathrm{E}} \right ) \exp\left ( -\frac{\tau\gamma_\mathrm{E}+\tau-1}{\overline{\gamma}_\mathrm{B}} \right )}} \nonumber \\ &\times \underset{I_3}{\underbrace{\left (  \tau\gamma_\mathrm{E}+\tau-1 \right )^q }}d\gamma_\mathrm{E}.
 \end{align}
 Here, with the help of~\cite[Eq.~(7.414.6)]{Gradshteyn} the value of the integral $I_2$ can be $\overline{\gamma}_\mathrm{E}$ when $n_\mathrm{E}=0$ or zero otherwise (i.e., $n_\mathrm{E}\not= 0$). In the former case, after by performing some algebraic manipulations, the first term of the SOP can be simplified as $C_1 k_\mathrm{B}!\overline{\gamma}_\mathrm{E} =1$. Next, by using~\cite[Eq.~(1.111)]{Gradshteyn}, $I_3$ can be expressed as
 \begin{align}\label{eq4ApendiceB}
I_3= &\sum_{a=0}^{q}\binom{q}{a}\left ( \tau-1 \right )^{q-a}\tau^a \exp\left ( -\frac{\tau-1}{\overline{\gamma}_\mathrm{B}} \right )  \underset{I_4}{\underbrace{\int_{0}^{\infty} L_{n_\mathrm{E}}\left ( \frac{\gamma_\mathrm{E}}{\overline{\gamma}_\mathrm{E}} \right ) }} \nonumber \\  & \times  \underset{I_4}{\underbrace{ \gamma_\mathrm{E}^a \exp\left (-\frac{\gamma_\mathrm{E}}{\overline{\gamma}_\mathrm{E}} -\frac{\tau\gamma_\mathrm{E}}{\overline{\gamma}_\mathrm{B}} \right )  d\gamma_\mathrm{E}.}}
\end{align}
Then, by solving the
corresponding integral in $I_4$, we get
 \begin{align}\label{eq5ApendiceB}
I_4= & \left ( \frac{1}{\overline{\gamma}_\mathrm{E}}+\frac{\tau}{\overline{\gamma}_\mathrm{B}} \right )^{-1-a}\Gamma\left ( 1+a \right )\nonumber \\  & \times { }_2F_1\left ( 1+a,-n_\mathrm{E},1,\frac{\overline{\gamma}_\mathrm{B}}{\overline{\gamma}_\mathrm{B}+\overline{\gamma}_\mathrm{E} \tau} \right ). 
\end{align}
Next, by combining~\eqref{eq1ApendiceB} to~\eqref{eq5ApendiceB}, the \text{SOP} can be formulated as in~\eqref{eq:SOPExact}, which concludes the proof.  

%% SOP approximation
%%%%%%%%%%%%%%%%%%%%%%%%%%%%%%%%%%%%%%%%%%%%%%%%%%%%%%%%%%%%%%%%%%%%%%%%%%%%%%%%%%%%%%%%%%%%%%%%%%%% SOP 
\subsection{$\mathrm{SOP}_{\mathrm{A}}$}
Substituting~\eqref{eq:SNR} into~\eqref{eq:sopL}, we get
\begin{align}\label{eq7ApendiceB}
 \text{SOP}_{\text{A}}=& 
\underset{C_1}{\underbrace{ \sum_{n_\mathrm{B}=0}^{\infty}C_{n_\mathrm{B}}\sum_{k_\mathrm{B}=0}^{n_\mathrm{B}}\frac{(-1)^{k_\mathrm{B}}}{k_\mathrm{B}!}\binom{n_\mathrm{B}}{k_\mathrm{B}} \frac{1}{\overline{\gamma}_\mathrm{E}} \sum_{n_\mathrm{E}=0}^{\infty} C_{n_\mathrm{E}}}} \nonumber \\ &\times \underset{I_5}{\underbrace{\int_{0}^{\infty} \exp\left ( -\frac{\gamma_\mathrm{E}}{\overline{\gamma}_\mathrm{E}} \right )L_{n_\mathrm{E}}\left ( \frac{\gamma_\mathrm{E}}{\overline{\gamma}_\mathrm{E}} \right ) }} \nonumber \\ &\times \underset{I_5}{\underbrace{\gamma \left ( k_\mathrm{B}+1,\frac{ \tau \gamma_\mathrm{E} }{\overline{\gamma}_\mathrm{B}} \right )d\gamma_\mathrm{E}.}}
 \end{align}
 Again, by using~\cite[Eq.~(8.352.1)]{Gradshteyn} into~\eqref{eq7ApendiceB},~$I_5$ can be reformulated  as
\begin{align}\label{eq8ApendiceB}
I_5=& \underset{I_5}{\underbrace{ k_\mathrm{B}! \int_{0}^{\infty} \exp\left ( -\frac{\gamma_\mathrm{E}}{\overline{\gamma}_\mathrm{E}} \right )L_{n_\mathrm{E}}\left ( \frac{\gamma_\mathrm{E}}{\overline{\gamma}_\mathrm{E}} \right ) d\gamma_\mathrm{E}}} \nonumber \\ &- k_\mathrm{B}! \sum_{q=0}^{k_\mathrm{B}}\frac{1}{q!}\left ( \frac{1}{\overline{\gamma}_\mathrm{B}} \right )^q \tau^q\underset{I_6}{\underbrace{\int_{0}^{\infty}\exp\left ( -\frac{\gamma_\mathrm{E}}{\overline{\gamma}_\mathrm{E}} \right ) }}\nonumber \\ &\times  \underset{I_6}{\underbrace{L_{n_\mathrm{E}}\left ( \frac{\gamma_\mathrm{E}}{\overline{\gamma}_\mathrm{E}} \right ) \exp\left ( -\frac{\tau\gamma_\mathrm{E}}{\overline{\gamma}_\mathrm{B}} \right )  \gamma_\mathrm{E}^q d\gamma_\mathrm{E}. }}
\end{align}
Here, note that $I_5$ is equivalent to $I_2$. Therefore, the first term of the $ \text{SOP}_{\text{A}}$ (i.e., $C_1 k_\mathrm{B}!\overline{\gamma}_\mathrm{E}$) once again equals unity, as discussed in the previous proof. On the other hand, by solving the corresponding integral in $I_6$, yields
\begin{align}\label{eq9ApendiceB}
I_6= & \left ( \frac{1}{\overline{\gamma}_\mathrm{E}}+\frac{\tau}{\overline{\gamma}_\mathrm{B}} \right )^{-1-q}\Gamma\left ( 1+q \right )\nonumber \\  & \times { }_2F_1\left ( -n_\mathrm{E},1+q,1,\frac{\overline{\gamma}_\mathrm{B}}{\overline{\gamma}_\mathrm{B}+\overline{\gamma}_\mathrm{E} \tau} \right ).
\end{align}
Finally, by combining~\eqref{eq7ApendiceB} to~\eqref{eq9ApendiceB}, the $\text{SOP}_{\text{A}}$ is reached  as in~\eqref{eq:SOPLower}, which completes the proof.

%%%%%%%%%%%%%%%%%%%%%%%%%%%%%%%%
%%%%%%%%%%%%%%%%%%%%%%%%%%%%%%%%
%APPENDIX A
%\newpage
\section{Proofs of Lemma~\ref{Lema2}}
\label{ap:SOPAsintotas}

%%%%%%%%% apendice C
%%%%%%%%%%%%%%%%%%%%%%%%%%%%%%%
\subsection{$\mathrm{SOP}^{\infty}$}
\subsubsection{Keeping $\overline{\gamma}_\mathrm{E}$ Fixed and $\overline{\gamma}_\mathrm{B}\rightarrow \infty$ }
In order to approximate~\eqref{eq:3} as $\overline{\gamma}_\mathrm{B}\rightarrow \infty$, we use the following relationship $\gamma \left (a,x \right )\approx x^s/s   $ as $x\rightarrow  0$. Therefore,~\eqref{eq:3} can be asymptotically expressed by
\begin{align}
\label{apenC:1}
F_\mathrm{B}(\gamma_\mathrm{B})\approx \sum_{n_\mathrm{B}=0}^{\infty}C_{n_\mathrm{B}}\sum_{k_\mathrm{B}=0}^{n_\mathrm{B}}\frac{(-1)^{k_\mathrm{B}}}{(k_\mathrm{B}+1)!}\binom{n_\mathrm{B}}{k_\mathrm{B}}\left (   \frac{\gamma_\mathrm{B}}{\overline{\gamma}_\mathrm{B} }  \right )^{k_\mathrm{B}+1}.
\end{align}
Substituting~\eqref{apenC:1} together with~\eqref{eq:2} into~\eqref{eq:sopL}, it follows that
\begin{align}\label{apenC:2}
 \mathrm{SOP}^{\infty}\approx&
  \sum_{n_\mathrm{B}=0}^{\infty}C_{n_\mathrm{B}}\sum_{k_\mathrm{B}=0}^{n_\mathrm{B}}\frac{(-1)^{k_\mathrm{B}}}{(k_\mathrm{B}+1)!}\binom{n_\mathrm{B}}{k_\mathrm{B}} \left (   \frac{1}{\overline{\gamma}_\mathrm{E} }  \right )\left (   \frac{\tau}{\overline{\gamma}_\mathrm{B} }  \right )^{k_\mathrm{B}+1} \nonumber \\ & \times  \sum_{n_\mathrm{E}=0}^{\infty} C_{n_\mathrm{E}} \underset{I_7}{\underbrace{\int_{0}^{\infty}\gamma_\mathrm{E}^{k_\mathrm{B}+1} \exp\left ( -\frac{\gamma_\mathrm{E}}{\overline{\gamma}_\mathrm{E}} \right )L_{n_\mathrm{E}}\left ( \frac{\gamma_\mathrm{E}}{\overline{\gamma}_\mathrm{E}} \right ) d\gamma_\mathrm{E} }}
 \end{align}
Next, with the aid of~\cite[Eq.~(7.414.7)]{Gradshteyn} to solve the integral in $I_7$,
the asymptotic SOP can be expressed as in~\eqref{eq:SOPAsin2}, which concludes the proof.
\subsubsection{Both $\overline{\gamma}_\mathrm{E}\rightarrow \infty$, $\overline{\gamma}_\mathrm{B}\rightarrow \infty$, and Fixed Ratio $\overline{\gamma}_\mathrm{E}/\overline{\gamma}_\mathrm{B}$}
From~\eqref{apenC:1}, the asymptotic PDF of $\mathrm{E}$ (i.e, $\overline{\gamma}_\mathrm{E}\rightarrow \infty$ ) is given by
\begin{align}
\label{apenC:3}
f_\mathrm{E}(\gamma_\mathrm{E})\approx \sum_{n_\mathrm{E}=0}^{\infty}C_{n_\mathrm{E}}\sum_{k_\mathrm{E}=0}^{n_\mathrm{E}}\frac{(-1)^{k_\mathrm{E}}}{k_\mathrm{E}!}\binom{n_\mathrm{E}}{k_\mathrm{E}}\left (   \frac{1}{\overline{\gamma}_\mathrm{E} }  \right )^{k_\mathrm{E}+1}  \gamma_\mathrm{E}^{k_\mathrm{E}}.
\end{align}
Substituting~\eqref{apenC:1} and~\eqref{apenC:3} into~\eqref{eq:sopL}, we have
\begin{align}
\label{apenC:4}
 \mathrm{SOP}^{\infty}\approx &
\sum_{n_\mathrm{B}=0}^{\infty}C_{n_\mathrm{B}}\sum_{k_\mathrm{B}=0}^{n_\mathrm{B}}\frac{(-1)^{k_\mathrm{B}}}{(k_\mathrm{B}+1)!}\binom{n_\mathrm{B}}{k_\mathrm{B}}\left (   \frac{\tau}{\overline{\gamma}_\mathrm{B} }  \right )^{k_\mathrm{B}+1} \nonumber \\
&\times \sum_{n_\mathrm{E}=0}^{\infty}C_{n_\mathrm{E}}\sum_{k_\mathrm{E}=0}^{n_\mathrm{E}}\frac{(-1)^{k_\mathrm{E}}}{k_\mathrm{E}!}\binom{n_\mathrm{E}}{k_\mathrm{E}}\left (   \frac{1}{\overline{\gamma}_\mathrm{E} }  \right )^{k_\mathrm{E}+1}  \nonumber \\
& \times  \underset{I_8}{\underbrace{\int_{0}^{\infty}\gamma_\mathrm{E}^{k_\mathrm{B}+k_\mathrm{E}+1} d\gamma_\mathrm{E} }}.
\end{align}
To solve $I_8$, one can rewrite it as
\begin{align}
\label{apenC:5}
 I_8=& \int_{0}^{\infty}\exp \left ( -\gamma_\mathrm{E} \right ) f\left ( \gamma_\mathrm{E} \right )  d\gamma_\mathrm{E},
\end{align}
where $f\left ( \gamma_\mathrm{E} \right )=\exp\left ( \gamma_\mathrm{E} \right ) \gamma_\mathrm{E}^{k_\mathrm{B}+k_\mathrm{E}+1}
$. Now, according to the Gauss-Laguerre quadrature
method~\cite[Eq.~(25.4.45)]{Abramowitz}, $I_8$ can be closely approximated by a weighted sum as
\begin{align}
\label{apenC:6}
 I_8\approx& \sum_{i=1}^{n}w_i f\left ( l_i \right ),
\end{align}
in which $l_i$ is the $i$th zero of
the Laguerre polynomial $L_n(\gamma_\mathrm{E})$ ~\cite[Eq.~(22.2.13)]{Abramowitz}, and $w_i=l_i\left [ (n+1)L_{n+1}\left ( l_i \right ) \right ]^{-2}$.\footnote{Experience in tests carried out show that, for large values of $n$ of the Laguerre polynomial, it becomes computationally hard to obtain a solution of the $\mathrm{SOP}^{\infty}$. However, choosing $n = 2$ is a good rule of thumb that has proven to be highly accurate with little computational effort.}This completes the proof.

%%%%%%%%%%%%%%%%%%%%%%%%%%%%%%%%%%%%%%%%%%%%%%%%%%%%%%%%%%%%%%%%%%%%%%%%%%%%%%%%
%%%%%%%%%%%%%%%%%%%%%%%%%%%%%%%%
%BIBLIOGRAPHY

\end{document}